\newif\ifdoublecol

\doublecoltrue
\documentclass[twocolumn]{IEEEtran}
\IEEEoverridecommandlockouts

\usepackage{etex} 

\ifCLASSINFOpdf
  \usepackage[pdftex]{graphicx}

  \usepackage{epstopdf}
\else
  \usepackage[dvips]{graphicx}
  \graphicspath{{../eps/}}
  \DeclareGraphicsExtensions{.eps}
\fi
\usepackage{cite}

\usepackage[cmex10]{amsmath}
\usepackage{xcolor,colortbl}
\definecolor{col1}{HTML}{3891A6}
\definecolor{col2}{HTML}{EF5B5B}
\definecolor{col3}{HTML}{3DDC97}

\usepackage{amsmath,amsthm,relsize}
\usepackage{amsfonts, amssymb, cuted}
\usepackage{mathtools}
\usepackage{algorithm}
\usepackage[noend]{algpseudocode}
\usepackage{cite}
\usepackage{soul}
\usepackage{tikz}
\usepackage{pgfplotstable}
\usepackage{pgfplots}
\usepackage{nicefrac}
\usepackage{enumitem}
\usepackage{support-caption}
\usepackage{subcaption}
\usepackage[font=footnotesize]{caption}
\usepackage{multirow}
\pgfplotsset{compat=1.15}
\usepackage{relsize}
\usetikzlibrary{patterns}
\usepackage{acro}
\usepackage{pifont}
\newtheorem{theorem}{Theorem}

\usepackage{todonotes}

\usepackage{mathtools}

\usepackage{array}
\usepackage{booktabs}

\usepackage{diagbox}
\usepackage{smartdiagram}
\usesmartdiagramlibrary{additions}
\usepackage{bm}
\usepackage{multicol}

\usepackage{tabularx}
\usepackage{colortbl}

\usepackage{hyperref}
\hypersetup{
    colorlinks=true,
    linkcolor=blue,
    citecolor=blue,
    filecolor=blue,
    urlcolor=blue
}

\usetikzlibrary{plotmarks}
  \usetikzlibrary{arrows.meta}
  \usepgfplotslibrary{patchplots}
  \usepackage{grffile}
  \pgfplotsset{plot coordinates/math parser=false}
  \newlength\figureheight
  \newlength\figurewidth
   \pgfplotsset{compat=1.11,
    /pgfplots/ybar legend/.style={
    /pgfplots/legend image code/.code={%
       \draw[##1,/tikz/.cd,yshift=-0.25em]
        (0cm,0cm) rectangle (3em,8pt);},
   },
}

\usepgfplotslibrary{groupplots,units}
\pgfplotsset{
  compat=1.9,
  unit code/.code 2 args={\si{#1#2}} 
}
\usepackage{siunitx}

\usepackage{color}
\usepackage{colortbl}
\usepackage{multirow}

\newlength{\Oldarrayrulewidth}

\definecolor{intnull}{RGB}{213,229,255}
\definecolor{inteins}{RGB}{128,179,255}
\definecolor{intzwei}{RGB}{42,127,255}
\definecolor{intdrei}{RGB}{0,85,212}
\definecolor{intvier}{RGB}{0,51,128}
\definecolor{intfunf}{RGB}{0,34,85}
\usepackage{arydshln} 

\usetikzlibrary{decorations.pathreplacing}
\renewcommand*{\arraystretch}{.4}

\newtheorem{lemma}{Lemma}

\newtheorem{remarknum}{Remark} 
\usepackage{arydshln}

\newcommand{\vbar}{\raisebox{.17ex}{\rule{.04em}{1.35ex}}}
\newcommand{\vbarind}{\raisebox{.01ex}{\rule{.04em}{1.1ex}}}
\newcommand{\R}{\ifmmode{\rm I}\hspace{-.2em}{\rm R} \else ${\rm I}\hspace{-.2em}{\rm R}$ \fi}
\newcommand{\T}{\ifmmode{\rm I}\hspace{-.2em}{\rm T} \else ${\rm I}\hspace{-.2em}{\rm T}$ \fi}
\newcommand{\N}{\ifmmode{\rm I}\hspace{-.2em}{\rm N} \else \mbox{${\rm I}\hspace{-.2em}{\rm N}$} \fi}
\newcommand{\B}{\ifmmode{\rm I}\hspace{-.2em}{\rm B} \else \mbox{${\rm I}\hspace{-.2em}{\rm B}$} \fi}
\newcommand{\Hil}{\ifmmode{\rm I}\hspace{-.2em}{\rm H} \else \mbox{${\rm I}\hspace{-.2em}{\rm H}$} \fi}
\newcommand{\C}{\ifmmode\hspace{.2em}\vbar\hspace{-.31em}{\rm C} \else \mbox{$\hspace{.2em}\vbar\hspace{-.31em}{\rm C}$} \fi}
\newcommand{\Cind}{\ifmmode\hspace{.2em}\vbarind\hspace{-.25em}{\rm C} \else \mbox{$\hspace{.2em}\vbarind\hspace{-.25em}{\rm C}$} \fi}
\newcommand{\Q}{\ifmmode\hspace{.2em}\vbar\hspace{-.31em}{\rm Q} \else \mbox{$\hspace{.2em}\vbar\hspace{-.31em}{\rm Q}$} \fi}
\newcommand{\Z}{\ifmmode{\rm Z}\hspace{-.28em}{\rm Z} \else ${\rm Z}\hspace{-.28em}{\rm Z}$ \fi}

\newtheorem{exmp}{Example}
\theoremstyle{definition}

\newcommand{\CB}[0]{{\mathcal{B}}}
\newcommand{\CC}[0]{{\mathcal{C}}}

\newcommand{\CK}[0]{{\mathcal{K}}}

\newcommand{\CS}[0]{{\mathcal{S}}}
\newcommand{\CT}[0]{{\mathcal{T}}}
\newcommand{\CU}[0]{{\mathcal{U}}}

\newcommand{\CW}[0]{{\mathcal{W}}}
\newcommand{\CX}[0]{{\mathcal{X}}}

\newcommand{\Bh}[0]{{\mathbf{h}}}

\newcommand{\Bw}[0]{{\mathbf{w}}}
\newcommand{\Bx}[0]{{\mathbf{x}}}
\newcommand{\By}[0]{{\mathbf{y}}}
\newcommand{\Bz}[0]{{\mathbf{z}}}

\newcommand{\BH}[0]{{\mathbf{H}}}
\newcommand{\BI}[0]{{\mathbf{I}}}

\newcommand{\BU}[0]{{\mathbf{U}}}

\newcommand{\SfS}[0]{{\mathsf{S}}}
\newcommand{\SfT}[0]{{\mathsf{T}}}

\newcommand{\SfX}[0]{{\mathsf{X}}}

\usepackage{tikz}

\usepackage{tikz}
\usetikzlibrary{arrows,
                calc,
                decorations.pathreplacing,
                calligraphy,
                matrix,
                positioning
                }

\DeclareAcronym{ADMM}{
    short = ADMM,
    long = alternating direction method of multipliers,
    list = Alternating Direction Method of Multipliers,
    tag = abbrev
}

\DeclareAcronym{AoA}{
    short = AoA,
    long = angle-of-arrival,
    list = Angle-of-Arrival,
    tag = abbrev
}

\DeclareAcronym{SISO}{
    short = SISO,
    long = single-input single-output,
    list = single-input single-output,
    tag = abbrev
}

\DeclareAcronym{MRT}{
    short = MRT,
    long = maximum ratio transmitter,
    list = maximum ratio transmitter,
    tag = abbrev
}

\DeclareAcronym{PDA}{
    short = PDA,
    long = placement delivery array,
    list = placement delivery array,
    tag = abbrev
}

\DeclareAcronym{EE}{
    short = EE,
    long = energy efficiency,
    list = energy efficiency,
    tag = abbrev
}

\DeclareAcronym{MDS}{
    short = MDS,
    long = maximum distance separation,
    list = maximum distance separation,
    tag = abbrev
}

\DeclareAcronym{SIC}{
    short = SIC,
    long = successive-interference-cancellation,
    list = successive-interference-cancellation,
    tag = abbrev
}

\DeclareAcronym{MAC}{
    short = MAC,
    long = multiple-access-channel,
    list = multiple-access-channel,
    tag = abbrev
}

\DeclareAcronym{AoD}{
    short = AoD,
    long = angle-of-departure,
    list = Angle-of-Departure,
    tag = abbrev
}

\DeclareAcronym{BB}{
    short = BB,
    long = base band,
    list = Base Band,
    tag = abbrev
}

\DeclareAcronym{BC}{
    short = BC,
    long = broadcast channel,
    list = Broadcast Channel,
    tag = abbrev
}

\DeclareAcronym{BS}{
    short = BS,
    long = base station,
    list = Base Station,
    tag = abbrev
}

\DeclareAcronym{BR}{
    short = BR,
    long = best response,
    list = Best Response, 
    tag = abbrev
}

\DeclareAcronym{CB}{
    short = CB,
    long = coordinated beamforming,
    list = Coordinated Beamforming,
    tag = abbrev
}

\DeclareAcronym{CC}{
    short = CC,
    long = coded caching,
    list = Coded Caching,
    tag = abbrev
}

\DeclareAcronym{CE}{
    short = CE,
    long = channel estimation,
    list = Channel Estimation,
    tag = abbrev
}

\DeclareAcronym{CoMP}{
    short = CoMP,
    long = coordinated multi-point transmission,
    list = Coordinated Multi-Point Transmission,
    tag = abbrev
}

\DeclareAcronym{CRAN}{
    short = C-RAN,
    long = cloud radio access network,
    list = Cloud Radio Access Network,
    tag = abbrev
}

\DeclareAcronym{CSE}{
    short = CSE,
    long = channel specific estimation,
    list = Channel Specific Estimation,
    tag = abbrev
}

\DeclareAcronym{CSI}{
    short = CSI,
    long = channel state information,
    list = Channel State Information,
    tag = abbrev
}

\DeclareAcronym{CSIT}{
    short = CSIT,
    long = channel state information at the transmitter,
    list = Channel State Information at the Transmitter,
    tag = abbrev
}

\DeclareAcronym{CU}{
    short = CU,
    long = central unit,
    list = Central Unit,
    tag = abbrev
}

\DeclareAcronym{D2D}{
    short = D2D,
    long = device-to-device,
    list = Device-to-Device,
    tag = abbrev
}

\DeclareAcronym{DE-ADMM}{
    short = DE-ADMM,
    long = direct estimation with alternating direction method of multipliers,
    list = Direct Estimation with Alternating Direction Method of Multipliers,
    tag = abbrev
}

\DeclareAcronym{DE-BR}{
    short = DE-BR,
    long = direct estimation with best response,
    list = Direct Estimation with Best Response,
    tag = abbrev
}

\DeclareAcronym{DE-SG}{
    short = DE-SG,
    long = direct estimation with stochastic gradient,
    list = Direct Estimation with Stochastic Gradient,
    tag = abbrev
}

\DeclareAcronym{DFT}{
	short = DFT,
	long = discrete fourier transform,
	list = Discrete Fourier Transform,
	tag = abbrev
}

\DeclareAcronym{DoF}{
    short = DoF,
    long = degrees of freedom,
    list = Degrees of Freedom,
    tag = abbrev
}

\DeclareAcronym{DL}{
    short = DL,
    long = downlink,
    list = Downlink,
    tag = abbrev
}

\DeclareAcronym{GD}{
	short = GD, 
	long = gradient descent,
	list = Gradeitn Descent,
	tag = abbrev
}

\DeclareAcronym{IBC}{
    short = IBC,
    long = interfering broadcast channel,
    list = Interfering Broadcast Channel,
    tag = abbrev
}

\DeclareAcronym{i.i.d.}{
    short = i.i.d.,
    long = independent and identically distributed,
    list = Independent and Identically Distributed,
    tag = abbrev
}

\DeclareAcronym{JP}{
    short = JP,
    long = joint processing,
    list = Joint Processing,
    tag = abbrev
}

\DeclareAcronym{KKT}{
    short = KKT,
    long = Karush-Kuhn-Tucker,
    tag = abbrev
}

\DeclareAcronym{LOS}{
	short = LOS,
	long = line-of-sight,
	list = Line-of-Sight,
	tag = abbrev
}

\DeclareAcronym{LS}{
    short = LS,
    long = least squares,
    list = Least Squares,
    tag = abbrev
}

\DeclareAcronym{LTE}{
    short = LTE,
    long = Long Term Evolution,
    tag = abbrev
}

\DeclareAcronym{LTE-A}{
    short = LTE-A,
    long = Long Term Evolution Advanced,
    tag = abbrev
}

\DeclareAcronym{MIMO}{
    short = MIMO,
    long = multiple-input multiple-output,
    list = Multiple-Input Multiple-Output,
    tag = abbrev
}

\DeclareAcronym{MISO}{
    short = MISO,
    long = multiple-input single-output,
    list = Multiple-Input Single-Output,
    tag = abbrev
}

\DeclareAcronym{MSE}{
    short = MSE,
    long = mean-squared error,
    list = Mean-Squared Error,
    tag = abbrev
}

\DeclareAcronym{MMSE}{
    short = MMSE,
    long = minimum mean-squared error,
    list = Minimum Mean-Squared Error,
    tag = abbrev
}

\DeclareAcronym{mmWave}{
	short = mmWave,
	long = millimeter wave,
	list = Millimeter Wave,
	tag = abbrev
}

\DeclareAcronym{MU-MIMO}{
    short = MU-MIMO,
    long = multi-user \ac{MIMO},
    list = Multi-User \ac{MIMO},
    tag = abbrev
}

\DeclareAcronym{OTA}{
    short = OTA,
    long = over-the-air,
    list = Over-the-Air,
    tag = abbrev
}

\DeclareAcronym{PSD}{
    short = PSD,
    long = positive semidefinite,
    list = Positive Semidefinite,
    tag = abbrev
}

\DeclareAcronym{QoS}{
	short = QoS,
	long = quality of service,
	list = Quality of Service,
	tag = abbrev
}

\DeclareAcronym{RCP}{
	short = RCP,
	long = remote central processor,
	list = Remote Central Processor,
	tag = abbrev
}

\DeclareAcronym{RRH}{
    short = RRH,
    long = remote radio head,
    list = Remote Radio Head,
    tag = abbrev
}

\DeclareAcronym{RSSI}{
    short = RSSI,
    long = received signal strength indicator,
    list = Received Signal Strength Indicator,
    tag = abbrev
}

\DeclareAcronym{RX}{
	short = RX,
	long = receiver,
	list = Receiver,
	tag = abbrev
}

\DeclareAcronym{SCA}{
    short = SCA,
    long = successive convex approximation,
    list = Successive Convex Approximation,
    tag = abbrev
}

\DeclareAcronym{SG}{
    short = SG,
    long = stochastic gradient,
    list = Stochastic Gradient,
    tag = abbrev
}

\DeclareAcronym{SNR}{
    short = SNR,
    long = signal-to-noise ratio,
    list = Signal-to-Noise Ratio,
    tag = abbrev
}

\DeclareAcronym{SINR}{
    short = SINR,
    long = signal-to-interference-plus-noise ratio,
    list = Signal-to-Interference-plus-Noise Ratio,
    tag = abbrev
}

\DeclareAcronym{SOCP}{
	short = SOCP, 
	long = second order cone program,
	list = Second Order Cone Program,
	tag = abbrev
}

\DeclareAcronym{SSE}{
    short = SSE,
    long = stream specific estimation,
    list = Stream Specific Estimation,
    tag = abbrev
}

\DeclareAcronym{SVD}{
	short = SVD,
	long = singular value decomposition,
	list = Singular Value Decomposition,
	tag = abbrev
}

\DeclareAcronym{TDD}{
	short = TDD,
	long = time division duplex,
	list = Time Division Duplex,
	tag = abbrev
}

\DeclareAcronym{TX}{
	short = TX,
	long = transmitter,
	list = Transmitter,
	tag = abbrev
}

\DeclareAcronym{UE}{
    short = UE,
    long = user equipment,
    list = User Equipment,
    tag = abbrev
}

\DeclareAcronym{UL}{
    short = UL,
    long = uplink,
    list = Uplink,
    tag = abbrev
}

\DeclareAcronym{ULA}{
	short = ULA,
	long = uniform linear array,
	list = Uniform Linear Array,
	tag = abbrev
}

\DeclareAcronym{UPA}{
    short = UPA,
    long = uniform planar array,
    list = Uniform Planar Array,
    tag = abbrev
}

\DeclareAcronym{WMMSE}{
    short = WMMSE,
    long = weighted minimum mean-squared error,
    list = Weighted Minimum Mean-Squared Error,
    tag = abbrev
}

\DeclareAcronym{WMSEMin}{
    short = WMSEMin,
    long = weighted sum \ac{MSE} minimization,
    list = Weighted sum \ac{MSE} Minimization,
    tag = abbrev
}

\DeclareAcronym{WBAN}{
	short = WBAN,
	long = wireless body area network,
	list = Wireless Body Area Network,
	tag = abbrev
}

\DeclareAcronym{WSRMax}{
    short = WSRMax,
    long = weighted sum rate maximization,
    list = Weighted Sum Rate Maximization,
    tag = abbrev
}

\usepackage{makecell}

 \usepackage{booktabs}

\begin{document}

\bstctlcite{IEEEexample:BSTcontrol}

\title{Low-Subpacketization MIMO Coded Caching with Flexible Stream Allocation}




\author{\IEEEauthorblockN{Mohammad NaseriTehrani,~\IEEEmembership{Student~Member,~IEEE,}  
MohammadJavad Salehi,~\IEEEmembership{Member,~IEEE,}~and Antti T\"olli},~\IEEEmembership{Senior~Member,~IEEE}
\thanks{
The authors are affiliated with the University of Oulu, Finland.Emails: 
\{mohammad.naseritehrani, mohammadjavad.salehi, antti.tolli\}@oulu.fi.
}
}



\maketitle


\begin{abstract}

Subpacketization remains a major obstacle to the practical deployment of coded caching (CC) in multi-antenna wireless networks. In this paper, we propose a low-complexity multiple-input multiple-output (MIMO) CC scheme that enables flexible delivery rate adaptation 
while substantially reducing subpacketization requirements. 
The proposed design 
\color{black}employs a virtual broadcast-channel decomposition
to align cache-profile placement with flexible MIMO stream elevation \color{black}
which extends the shared-cache model to multi-antenna receivers. This enables 
flexible selection of feasible user and stream configurations, thereby providing explicit control over the spatial multiplexing gain under linear decodability constraints. 
Analytical results show that the proposed framework can asymptotically approach the best-known achievable degrees of freedom (DoF) under linear decodability constraints  while requiring orders-of-magnitude lower subpacketization than existing schemes. Numerical evaluations further demonstrate that this flexibility yields notable throughput improvements at practical signal-to-noise ratios.

\end{abstract}

\begin{IEEEkeywords}
\noindent coded caching, multicasting, MIMO communications, Degrees of freedom
\end{IEEEkeywords}

\section{Introduction}

\color{black}Coded caching~(CC) exploits user-side storage~\cite{maddah2014fundamental}
to create coded-multicast opportunities during delivery, thereby improving
spectral efficiency. Originally developed for single-input single-output~(SISO)
broadcast channels (BC), CC was extended to MISO systems, where spatial multiplexing
and global caching gains jointly enhance the achievable degrees of freedom~(DoF)
~\cite{shariatpanahi2018physical}.
Finite-SNR precoder
design has also attracted increasing attention in MISO CC~\cite{tolli2017multi}
and related multi-stream physical-layer designs~\cite{yang2024oam,ma2026oam},
highlighting practical performance beyond DoF. \color{black}
%
%
Extending CC to multiple-input multiple-output~(MIMO) networks introduces
additional design opportunities and implementation challenges. Prior works
characterized the fundamental DoF limits of cache-aided MIMO systems%
~\cite{cao2019treating}
and developed low-complexity
transmission schemes for large transmit-side antenna configurations%
~\cite{salehi2021MIMO}, while our recent works%
~\cite{tehrani2024enhanced,naseritehrani2026cache,naseritehrani2026cacheasymmetric}
established improved single-shot achievable DoF bounds under linear decodability constraints.

A major challenge in practical CC implementations is the rapidly increasing subpacketization level, which reflects the number of file segments required for cache placement and delivery. In classical CC schemes, subpacketization grows exponentially with the number of users, thereby severely limiting practical scalability. Several low-subpacketization designs have been proposed for MISO CC systems, including shared-cache schemes~\cite{lampiris2018adding}
and cyclic constructions~\cite{salehi2020lowcomplexity}. These approaches preserve the optimal DoF performance of MISO CC while significantly reducing subpacketization. However, they mainly focus on achieving the maximum possible DoF, rather than enabling the transmission scheme to be tuned for optimized finite-SNR performance. Moreover, they do not exploit the receive-side spatial multiplexing capabilities available in more general MIMO settings.


\color{black} 
In this paper, we develop a low-subpacketization MIMO-CC framework
that integrates shared-cache profile placement with flexible
spatial multiplexing at the transceiver. The proposed construction maps the original
MIMO broadcast channel to a virtual MISO (V-MISO) network and links
the virtual transmit dimension of this network with a group-aligned
cache placement inspired by the shared-cache paradigm, under which users within the same virtual group are
assigned identical cache contents. 
Consequently, the same
virtual-group size 
simultaneously governs
the cache-profile structure, the
virtual transmit dimension, and the spatial interference structure
after MIMO elevation.

This design generalizes the fixed low-subpacketization elevation
in~\cite{salehi2021MIMO}, where the virtual transmit dimension is
determined by the antenna configuration, each user employs the
maximum receive-side spatial multiplexing order, and the corresponding
transmit-to-receive antenna ratio is required to be integer. By
relaxing this fixed operating structure and allowing both quantities
to vary subject to linear decodability, the proposed construction
extends the low-subpacketization design beyond the restrictive
applicability regime of~\cite{salehi2021MIMO}.
Relative to the broader flexible MIMO-CC framework
in~\cite{naseritehrani2026cache}, whose file segmentation
is tied to the total number of users, the proposed profile-domain construction
substantially reduces subpacketization. At the same time, it enables
flexible tuning of both the number of simultaneously served users and
the number of spatial streams per user within a structured set of
feasible operating points, thereby reducing the operating-point search
while retaining spatial-multiplexing flexibility. The resulting
DoF-subpacketization operating points can approach the best-known
achievable DoF under linear decodability, while at finite SNR the
per-user stream multiplicity can be selected according to the
achievable symmetric rate.
\color{black}


\textbf{Notations.} Bold upper- and lower-case letters denote matrices
and vectors, respectively. $[\Bh_k]_{k\in\CK}$ and
$[\BH_k]_{k\in\CK}$ denote the horizontal concatenations of
$\Bh_k$ and $\BH_k$, respectively. Calligraphic letters denote sets,
$|\CK|$ denotes cardinality, $\CK\backslash b$ excludes $b$ from
set $\CK$, and for an integer $K$, $[K]\triangleq\{1,\ldots,K\}$.

\section{System Model}
\label{section:sys_model}

We consider a cache-aided MIMO 
BC in which a base station (BS) with $L$ transmit antennas serves $K$ users (UEs), each with $G$ receive antennas. The system operates over a library $\CW$ of $N$ equal-sized files $\{W_1,\dots,W_N\}$. Each UE has a cache of size $M$ files. The global caching gain is defined as $t = \frac{KM}{N}$, representing the aggregate number of library copies stored across all UE caches. We assume $K \ge t+1$.
%
During placement, fragments of the files in $\CW$ are stored in UE caches without knowledge of future requests. At the beginning of the delivery phase, each UE~$k \in [K]$ reveals its requested file index $d_k$ to the BS. Delivery is organized as a sequence of multicast transmissions, each serving a UE subset $\CX \subseteq [K]$ of size $\Omega$, where $\Omega \in [t+1, K]$ is a design parameter controlling the number of simultaneously served UEs.

Let $\SfX$ denote the collection of UE subsets $\CX$ scheduled for transmission by the delivery algorithm. For each $\CX \in \SfX$, the BS transmits a signal vector $\Bx(\CX) \in \mathbb{C}^{L}$. The received signal at UE~$k \in \CX$ is
\begin{equation}
\By_k(\CX) = \BH_k(\CX)\mathbf{x}(\CX) + \Bz_k(\CX),
\end{equation}
where $\BH_k(\CX) \! \in \!\mathbb{C}^{G \times L}$ denotes a block-fading channel matrix with independent and identically distributed complex Gaussian entries of unit variance, and $\Bz_k(\CX)\! \sim \!\mathcal{CN}(\mathbf{0}, N_0\BI)$ is additive noise. Perfect channel state information is assumed at the transmitter. Each served UE~$k\! \in \!\CX$ applies a linear receive beamforming matrix $\BU_k(\CX)\! \in \! \mathbb{C}^{G \times \beta}$ to decode $\beta$ spatial streams, where $\beta \le G$. 
The parameters $(\Omega,\beta)$ jointly define the spatial
multiplexing configuration and serve as the main design variables.
%
During delivery, file fragments may be further split into multiple subpackets, with $\Theta$ denoting the resulting overall subpacketization.
Each transmission vector $\Bx(\CX)$ delivers a set of subpackets in parallel with an achievable max–min rate $R(\CX)$ (file/s), such that all UEs in $\CX$ can successfully decode their requested subpackets. 
%
The corresponding transmission and total delivery times are
$T(\CX)=\frac{1}{\Theta R(\CX)}$ and
$T_{\text{total}}=\sum_{\CX\in\SfX}T(\CX)$, respectively, yielding the symmetric rate
$R_{\text{sym}}=\frac{K}{T_{\text{total}}}$.

\section{The Proposed Low-Subpacketization Design}
\label{sec:virtual_bc_decomp}

%
To construct a low-subpacketization, flexible MIMO-CC scheme, we propose a virtual broadcast-channel decomposition framework. The key idea is to 
map the original cache-aided MIMO BC to an equivalent V-MISO network, design the placement and delivery in that virtual domain, and then elevate the resulting construction back to the original MIMO setting. 
The reduction in subpacketization is achieved through a group-aligned cache placement inspired by the shared-cache paradigm, under which virtual users (V-UEs) within the same virtual group are assigned identical cache contents.
\color{black}
Importantly, the same grouping also determines the spatial interference structure after elevation, thereby coupling the cache-profile design with the feasible spatial stream allocation.\color{black}
%
%
\begin{theorem}
\label{thm:main_virtual}
For any pair of positive integers $(\Omega,\beta)$ satisfying
\begin{equation}
t+1 \le \Omega \le t+L, \qquad \beta \le G, \qquad (\Omega - t)\beta \le L ,
\label{eq:ld_condition_clean_n}
\end{equation}
if the integer conditions
\begin{equation}
\frac{K}{\Omega-t} \in \mathbb{Z}_+,\qquad
\frac{t}{\Omega-t} \in \mathbb{Z}_+,
\label{eq:integer_conditions_n}
\end{equation}
are met, there exists a linearly decodable MIMO CC scheme that achieves a DoF of $\Omega \times \beta$ per transmission with subpacketization level
\vspace{-2mm}
\begin{equation}
\label{eq:subpacketization_main}
\Theta(\Omega,\beta)
=
\beta \binom{\frac{K}{\Omega-t}}{\frac{t}{\Omega-t}} .
\end{equation}
\end{theorem}
\begin{proof}
Let $(\Omega,\beta)$ satisfy~\eqref{eq:ld_condition_clean_n} 
and~\eqref{eq:integer_conditions_n}, and define $\gamma \triangleq \frac{M}{N}$ and $\Delta \triangleq \Omega-t$.
By~\eqref{eq:integer_conditions_n}, both $P = \frac{K}{\Delta}$ and $P\gamma = \frac{t}{\Delta}$ are integers. 
%
We proceed by constructing the placement and delivery algorithms using zero-forcing (ZF) beamformers,
and then establish linear decodability, achievable DoF, and subpacketization.

\noindent $\bullet$ \textbf{Placement algorithm.} It consists of two steps.

\textit{1) User grouping.}
Partition UEs into $P$ disjoint groups
\begin{equation}
\CU_p = \{\ell P + p : \ell = 0,1,\dots,\Delta-1\},
\qquad p \in [P],
\end{equation}
so that each group contains exactly $\Delta$ UEs.

\textit{2) Cache placement.}
All UEs within the same group share identical cache contents. Define the set of cache-profile indices
\begin{equation}
\SfT = \left\{ \CT \subseteq [P] : |\CT| = P\gamma \right\}.
\end{equation}
Each file $W_n$ is split into $|\SfT|=\binom{P}{P\gamma}$ subfiles $W_n = \{ W_n^{\CT} : \CT \in \SfT \}$, and for each $p \in [P]$, all UEs $k \in \CU_p$ are assigned the cache
\begin{equation}
\CC_k = \CC_{\CU_p} = \{ W_n^{\CT} : p \in \CT,\ n \in [N] \}.
\label{eq:cache_group}
\end{equation}
%
This reduces the effective placement dimension from $K$ to
$P=\frac{K}{\Delta}$ cache groups, reducing subpacketization%
~\cite{lampiris2018adding}.

\noindent $\bullet$ \textbf{Delivery algorithm.} It consists of three steps.

\textit{1) V-MISO network setup.}
Consider a V-MISO network with $K$ single-antenna V-UEs and a virtual BS (V-BS) with $\Delta$ antennas.
%
Assume V-UE $k$ 
inherits the cache contents and file request of UE $k$ in the
original MIMO network.
%

\textit{2) V-MISO transmission.}
Define the collection of active cache profile sets as
$\SfS = \left\{ \CS \subseteq [P] : |\CS| = P\gamma + 1 \right\}$.
Each $\CS \in \SfS$ corresponds to the V-UE subset
\begin{equation}
\CX(\CS) \triangleq \bigcup\nolimits_{p \in \CS} \CU_p .
\end{equation}
Since $|\CU_p|=\Delta$ and $|\CS|=P\gamma+1$, we obtain
$|\CX(\CS)| = \Delta(P\gamma+1)=t+\Delta=\Omega$,
i.e., each $\CS$ induces a V-UE subset of size $\Omega$. For each $\CS$, we create the V-MISO transmission
\begin{equation}
\hat{\Bx}(\CS) = \sum_{p \in \CS}\;\sum_{k \in \CU_p} W_{d_k}^{\CS\setminus p}\,
\Bw_{\CU_p\setminus k},
\label{eq:vmiso_tx_n}
\end{equation}
to deliver one stream to every V-UE $k \in \CX(\CS)$. In~\eqref{eq:vmiso_tx_n}, the virtual beamforming vector $\Bw_{\CU_p\setminus k} \in \mathbb{C}^{\Delta \times 1}$ lies in the null space of the stacked virtual channel matrix of the remaining V-UEs in the same group, i.e., 
\begin{equation}
\hat{\BH}_{\CU_p\setminus k} = \big[\hat{\Bh}_{k'}^{H}\big]_{k' \in \CU_p\setminus k}
\in \mathbb{C}^{(\Delta-1)\times \Delta} ,
\end{equation}
where $\hat{\Bh}_{k'} \in \mathbb{C}^{\Delta \times 1}$ represents the virtual channel between V-UE~$k'$ and the V-BS. 
From~\eqref{eq:vmiso_tx_n}, intra-group interference is nulled out by ZF beamforming, while the remaining inter-group interference is removed using 
cached side information, since for each V-UE \(k\in\CU_p\), all terms intended for other groups correspond to subfiles whose indices contain $p$, and are hence available in its cache.

\textit{\color{black}3\color{black}) Elevation to the MIMO channel.}
Each V-MISO transmission $\hat{\Bx}(\CS)$ is elevated to a real transmission vector for the original MIMO setup. In this regard, each virtual stream for a V-UE~$k$ in~\eqref{eq:vmiso_tx_n} is replaced by $\beta$ parallel streams for UE~$k$ in the MIMO network. Specifically, each subfile $W_{d_k}^{\CS\setminus p }$ is partitioned into $\beta$ equal-sized subpackets $W_{d_k}^{\CS\setminus p, (q)}$, $q\in[\beta]$, and the elevated MIMO transmission becomes
\begin{equation}
\Bx(\CX(\CS)) = \sum_{p \in \CS}\;\sum_{k \in \CU_p} \; \sum_{q=1}^{\beta}
W_{d_k}^{\CS\setminus p , (q)}\,
\Bw_{\CU_p\setminus k}^{(q)},
\label{eq:mimo_lift_tx_n}
\end{equation}
where $\Bw_{\CU_p\setminus k}^{(q)} \in \mathbb{C}^{L}$ denotes the ZF beamformer associated with the $q$-th stream (more explanation is provided below).

\noindent $\bullet$ \textbf{Linear decodability.}
Each beamformer $\Bw_{\CU_p\setminus k}^{(q)}$ in~\eqref{eq:mimo_lift_tx_n} is designed such that the stream $W_{d_k}^{\CS\setminus p,(q)}$ is nulled toward the remaining $\Delta-1$ UEs in the same active group after receive combining. This results in the elevated nulling matrix
\begin{equation}
\BH_{\CU_p\setminus k} = \big[\BU_{k'}^{H}(\CX(\CS))\BH_{k'}(\CX(\CS))\big]_{k' \in \CU_p\setminus k}.
\end{equation}
It can be easily seen that $\BH_{\CU_p\setminus k} \in \mathbb{C}^{(\Delta-1)\beta \times L}$.
To suppress the $(\Delta-1)\beta$ interfering streams while delivering $\beta$ desired streams, the transmit signal must lie in a subspace of dimension at least $\beta$, which requires
\[
L - (\Delta-1)\beta \ge \beta
\quad \Longleftrightarrow \quad
(\Omega - t)\beta \le L .
\]
\color{black}Moreover, resolving $\beta$ desired streams per receiver requires
$\beta\le G$; together, these conditions coincide with those in
~\eqref{eq:ld_condition_clean_n}.\color{black}

\noindent $\bullet$ \textbf{Achievable DoF.}
%
Since each 
subset $\CX(\CS)$ serves $\Omega$ UEs each decoding $\beta$ streams,
the achievable DoF per~transmission is
\begin{equation}
\mathrm{DoF}(\Omega,\beta) = \Omega \times \beta .
\end{equation}

\noindent $\bullet$ \textbf{Subpacketization.}
The placement phase splits each file into
\begin{equation}
|\SfT| = \binom{P}{P\gamma} = \binom{\frac{K}{\Delta}}{\frac{t}{\Delta}}
\end{equation}
subfiles,
while MIMO elevation step introduces an 
extra factor $\beta$ by further partitioning each virtual stream into $\beta$ subpackets. Hence, the total subpacketization is~\eqref{eq:subpacketization_main}, completing the proof.
\end{proof}
\begin{remarknum}
\label{subpacketization_ratio}
For large \(K\) and moderate $P$ (i.e., $P \ge 10$), if \(\Omega\) and $\beta$ are fixed
and $\Delta >1$, the subpacketization of the proposed scheme is exponentially smaller than that of the MIMO-CC scheme in~\cite{naseritehrani2026cache} by a factor of:
\begin{equation}
\label{eq:subpack_ratio}
\begin{aligned}
\frac{\beta \binom{P}{P\gamma}}
{\beta \binom{K}{t}\binom{K-t-1}{\Omega-t-1}}
\approx
\mathcal{O}\left(
K^{-(\Delta-1)}
\,2^{-K\left(1-\frac{1}{\Delta}\right)H(\gamma)}
\right),
\end{aligned}
\end{equation}
where $H(\gamma)$ is the binary entropy function. 
The ratio in~\eqref{eq:subpack_ratio} is calculated by using approximations
\begin{equation*}
    \begin{aligned}
        \binom{K}{K\gamma}
\approx
\frac{2^{K H(\gamma)}}{\sqrt{2\pi K\gamma(1-\gamma)}}, \quad \binom{P}{P\gamma}
\approx
\frac{2^{P H(\gamma)}}{\sqrt{2\pi P\gamma(1-\gamma)}},
    \end{aligned}
\end{equation*}
which follow the Stirling approximation $n! \approx\sqrt{2\pi n}(\frac{n}{e})^n$ and are valid for moderate $K$ and $P$ values. Moreover, we have used the asymptotic approximation 
\begin{equation*}
    \binom{K-K\gamma-1}{\Delta-1} \approx \frac{(K-K\gamma)^{\Delta-1}}{(\Delta-1)!}.
\end{equation*}
\end{remarknum}

\begin{exmp}\label{exmp1}
Consider a MIMO-CC network with $K=24$, $L=13$, $G=2$, and $\gamma = 0.5$ (hence, $t=12$). For this network, a subset of $(\Omega, \beta)$ pairs satisfying the conditions of Theorem~\ref{thm:main_virtual} and their respective subpacketization requirement with the proposed solution are shown in Table~\ref{tab:dof_theta_by_Omega_beta}.
W.l.o.g., let us focus on $(\Omega,\beta)=(18,2)$, with $\mathrm{DoF}=36$. 

\noindent\textbf{Cache placement:} Users are split into $P=K/(\Omega-t)=4$ disjoint groups of size $\Omega-t=6$,
as shown in Figure~\ref{fig:ISIT_sysmo}. 
Given $P\gamma=2$, we split each file $W_n$ into $\binom{P}{P\gamma}=6$ subfiles, i.e., $W_n \rightarrow \{ W_n^{\{1,2\}}, \cdots ,W_n^{\{3,4\}}\}$.  
Now, defining
\begin{equation*}
    \begin{array}{l}
    C_1=\{ W_n^{\{1,2\}}, \cdots ,W_n^{\{1,4\}}\},\: C_2=\{ W_n^{\{1,2\}}, \cdots ,W_n^{\{2,4\}}\},\\
    C_3=\{ W_n^{\{1,3\}}, \cdots ,W_n^{\{3,4\}}\},\: C_4=\{ W_n^{\{1,4\}}, \cdots ,W_n^{\{3,4\}}\},
    \end{array}
\end{equation*}
we put the subfiles in set $\CC_p$ in the cache memory of each user in group $\CU_p$.

\noindent\textbf{Data delivery:} Let us consider the profile set $\CS=\{1,2,3\}$. The V-MISO transmission vector for this set is given as:
\begin{equation*}
\begin{aligned}
\hat\Bx(\CX({\{1,2,3\}})) &= W_{d_{1}}^{\{2,3\}} \Bw_{\CU_1 \setminus 1}+\cdots +W_{d_{21}}^{\{2,3\}} \Bw_{\CU_1 \setminus 21}\\
&+ W_{d_{2}}^{\{1,3\}} \Bw_{\CU_2 \setminus 2}+\cdots +W_{d_{22}}^{\{1,3\}} \Bw_{\CU_2 \setminus 22}\\
 &+ W_{d_{3}}^{\{1,2\}} \Bw_{\CU_3 \setminus 3}+\cdots +W_{d_{23}}^{\{1,2\}} \Bw_{\CU_3 \setminus 23}.
\end{aligned}
\end{equation*}
Substituting each virtual stream in $\hat\Bx(\CX({\{1,2,3\}}))$ with $\beta=2$ streams, the real MIMO transmission vector is given as:
\begin{equation*}
\begin{aligned}
\Bx&(\CX({\{1,2,3\}})) = 
W_{d_{1}}^{\{2,3\},(1)} \Bw_{\CU_1 \setminus 1}^{(1)}+W_{d_{1}}^{\{2,3\},(2)} \Bw_{\CU_1 \setminus 1}^{(2)}\\
&\qquad +\cdots + W_{d_{21}}^{\{2,3\},(1)} \Bw_{\CU_1 \setminus 21}^{(1)}+W_{d_{21}}^{\{2,3\},(2)} \Bw_{\CU_1 \setminus 21}^{(2)}\\
& \qquad+ \cdots+W_{d_{3}}^{\{1,2\},(1)} \Bw_{\CU_3 \setminus 3}^{(1)}+W_{d_{3}}^{\{1,2\},(2)} \Bw_{\CU_3 \setminus 3}^{(2)} \\
&\qquad+\cdots + W_{d_{23}}^{\{1,2\},(1)} \Bw_{\CU_3 \setminus 23}^{(1)}+W_{d_{23}}^{\{1,2\},(2)} \Bw_{\CU_3 \setminus 23}^{(2)}.
\end{aligned}
\end{equation*}

\end{exmp}

\begin{figure}[t]
    \centering
    \includegraphics[width=0.75\columnwidth]{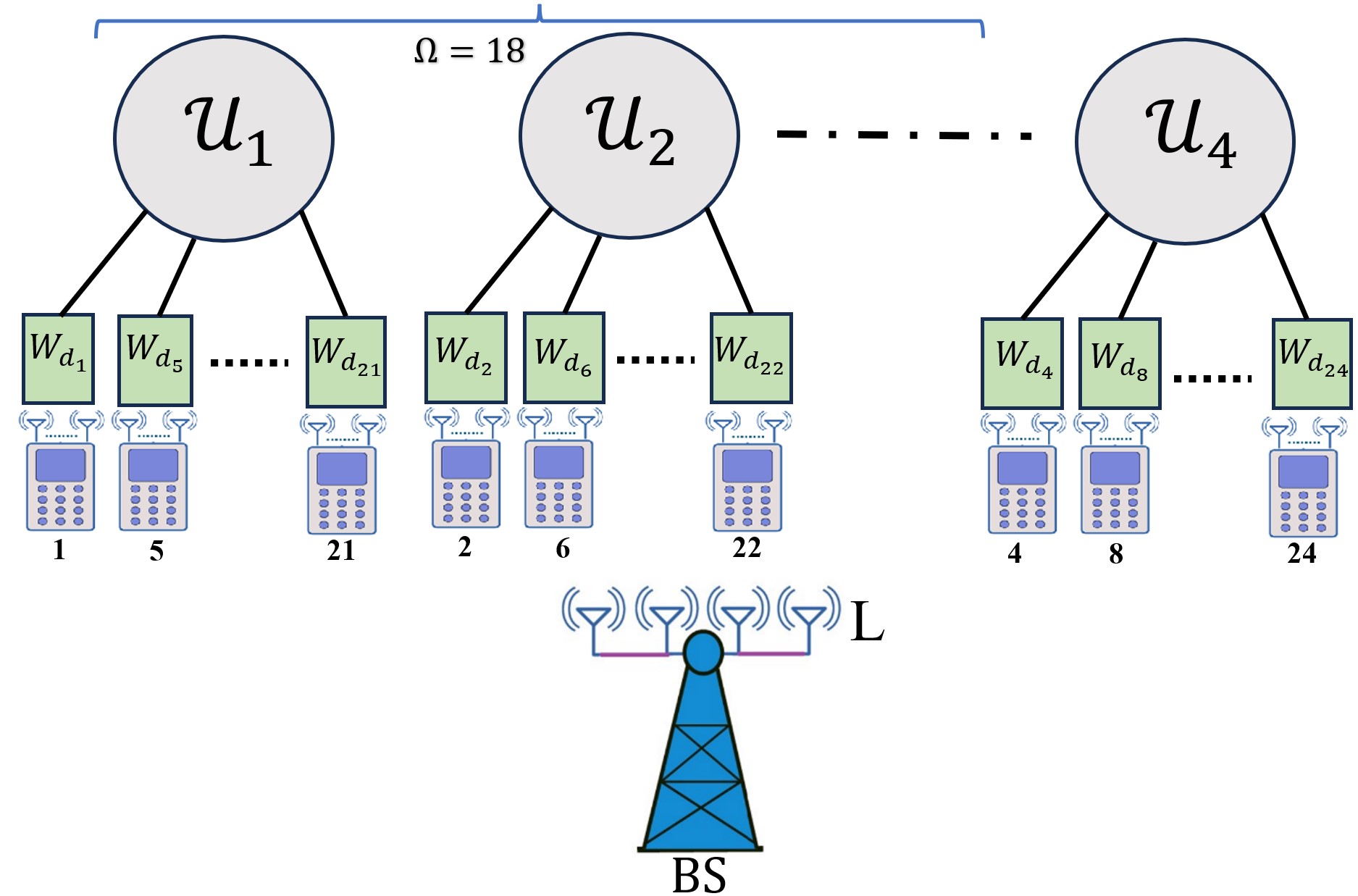} 
    \caption{MIMO setup of Example~\ref{exmp1} with $K=24$ users assigned to $P=4$ cache profiles, and $\Omega=18$ users served in each transmission.}
    \label{fig:ISIT_sysmo}
\end{figure}

\begin{table}[t]
\centering
\caption{Pairs of $(\Omega, \beta)$ satisfying Theorem~\ref{thm:main_virtual} conditions for the network in Example~\ref{exmp1}, and their corresponding subpacketization.}
\footnotesize
\setlength{\tabcolsep}{3.5pt}
\renewcommand{\arraystretch}{1.15}
\resizebox{.83\columnwidth}{!}{%
\begin{tabular}{|c|
c:c| c:c| c:c| c:c| c:c|}
\hline
\diagbox[dir=NW,width=5.8em,height=2.2em]{$\beta$}{\raisebox{-0.2em}{$\Omega$}}
& \multicolumn{2}{c|}{14}
& \multicolumn{2}{c|}{15}
& \multicolumn{2}{c|}{16}
& \multicolumn{2}{c|}{18}
& \multicolumn{2}{c|}{24} \\
\cline{2-11}
& $\Omega \beta$ & $\Theta$
& $\Omega \beta$ & $\Theta$
& $\Omega \beta$ & $\Theta$
& $\Omega \beta$ & $\Theta$
& $\Omega \beta$ & $\Theta$ \\
\hline
1 & 14 & 924  & 15 & 70  & 16 & 20 & 18 & 6  & 24 & 2 \\
\hdashline
2 & 28 & 1848 & 30 & 140 & 32 & 40 & 36 & 12 & -- & -- \\
\hline
\end{tabular}%
}
\label{tab:dof_theta_by_Omega_beta}
\end{table}

\begin{remarknum}[Non-integer parameter extension]
\label{rem:noninteger_extension}
\color{black}
The integer constraints in~\eqref{eq:integer_conditions_n} can be
relaxed using a standard combination of phantom-user padding
(cf.~\cite{lampiris2018adding}) and memory sharing. Specifically,
the system can first be padded to
$\bar K=\Delta\left\lceil\frac{K}{\Delta}\right\rceil$ UEs so that
$\frac{\bar K}{\Delta}\in\mathbb Z_+$, with
$0\leq\frac{\bar K-K}{K}<\frac{\Delta}{K}$.
Memory sharing can then be applied between the two nearest integer
values of the corresponding padded placement parameter
$\frac{\bar K}{\Delta}\gamma$, with the resulting delivery time given
by the corresponding weighted combination of the two delivery times.
The associated cache-ratio spacing is
$\frac{\Delta}{\bar K}$.
Thus, the finite-size departure from the exact integer construction is
governed by the relative padding mismatch, bounded by
$\frac{\Delta}{K}$, and the cache-ratio discretization
$\frac{\Delta}{\bar K}$, both of which vanish for fixed $\Delta$ as
$K$ grows. This preserves the same virtual decomposition and MIMO
elevation structure, while the resulting generalized subpacketization
remains of the same order as~\eqref{eq:subpacketization_main}.
\color{black}
\end{remarknum}


Let \(\CB\) denote the set of all feasible \((\Omega,\beta)\) pairs satisfying~\eqref{eq:ld_condition_clean_n} and~\eqref{eq:integer_conditions_n}. The maximum achievable DoF of the proposed low-subpacketization scheme is given as $\Omega^* \!\times\! \beta^*$, where 
\begin{equation}
    (\Omega^*,\beta^*) = \arg\max_{(\Omega,\beta) \in \CB} \Omega \times \beta.
\end{equation}

\begin{lemma}
\label{lem:dof_gap}
For a given MIMO-CC setup, if $L \ge G$, the gap between the largest DoF values achievable by the proposed scheme and the scheme in~\cite{naseritehrani2026cache} is at most $L-1$. In other words, if we show the optimized number of users served per transmission and the number of streams per user for the MIMO-CC scheme in~\cite{naseritehrani2026cache} by $\Omega_{\mathrm{opt}}$ and $\beta_{\mathrm{opt}}$, respectively, then
\begin{equation}
\label{eq:dof_gap}
    \Omega_{\mathrm{opt}} \times \beta_{\mathrm{opt}} - \Omega^* \times \beta^* \le L-1.
\end{equation}
\end{lemma}
\begin{proof}
    As shown in~\cite[Theorem~1]{naseritehrani2026cache},
for the MIMO-CC scheme to be linearly decodable, we should have
$\beta_{\rm opt}\leq G$ and
    \begin{equation}
        L-\beta_{\mathrm{opt}}(\Omega_{\mathrm{opt}}-t-1) \ge \left\lceil \frac{\beta_{\mathrm{opt}}}{\binom{\Omega_{\mathrm{opt}}-1}{t}} \right\rceil.
    \end{equation}
    %
    Since $\lceil \beta_{\mathrm{opt}}/\binom{\Omega_{\mathrm{opt}}-1}{t} \rceil \ge 1$, the latter condition implies $\beta_{\mathrm{opt}}(\Omega_{\mathrm{opt}}-t-1) \le L-1$,
    which yields
    \begin{equation}
    \label{eq:omega_beta_DoF_opt}
    \begin{aligned} 
        \hspace{-2mm}\Omega_{\mathrm{opt}} \times \beta_{\mathrm{opt}} \le L-1 + \beta_{\mathrm{opt}}(t+1) \le L-1+Gt+G .
    \end{aligned}
    \end{equation}
    On the other hand, for the proposed scheme, it can be easily verified that $(\Omega = t+1, \beta = G)$ is always a feasible point as it satisfies all the conditions in Theorem~\ref{thm:main_virtual}. This results in
    \begin{equation}
        \Omega^* \times \beta^* \ge G(t+1) = Gt+G,
    \end{equation}
    which, together with~\eqref{eq:omega_beta_DoF_opt} results in~\eqref{eq:dof_gap}. 
\end{proof}

It should be noted that the bound in Lemma~\ref{lem:dof_gap} can be tight. In particular, this can happen when the largest DoF of the proposed scheme occurs at $\Omega^* = t+1$ (which is the case, for example, when \(\gcd(K,t)=1\)). 
For instance, in a setup with $L=13$, $G=6$, $\gamma=0.0375$, $K=80$, and $t=3$, the best DoF values achievable by the MIMO-CC scheme~\cite{naseritehrani2026cache} and our proposed scheme are $\Omega_{\rm opt}\times \beta_{\rm opt}=6\times 6$ and $\Omega_{}^*\times \beta_{}^*=4\times 6$, respectively. Hence, the gap is \(36-24=12=L-1\).
%

Nevertheless, as shown in Section~\ref{section:Simulations}, the DoF gap is much smaller than \(L-1\) in most cases, and can even vanish.

\begin{remarknum}[Beamformer design]\label{rem:Bf_design}
\color{black}
The proof of Theorem~\ref{thm:main_virtual} employs ZF beamformers to establish the achievable DoF in a tractable manner through spatial dimension counting. At practical finite SNR, however, the Tx/Rx
beamformers can be optimized by solving 
\vspace{-2.mm}
\begin{equation}
\label{eq:finite_snr_bf}
\begin{aligned}
&\max\limits_{\{\Bw_{\CU_p\setminus k}^{(q)}\},\,\{\BU_k(\CX)\}}\quad
\min\limits_{k\in\CX}\frac{1}{\beta}
\sum\nolimits_{q=1}^{\beta}
\log_2\!\bigl(1+\Gamma_{k,q}(\CX)\bigr)
\\[-2.0mm]
&\mathrm{s.t.}\:
\sum\nolimits_{p\in \CS}\sum\nolimits_{k\in\CU_p}
\sum\nolimits_{q=1}^{\beta}
\left\|\Bw_{\CU_p\setminus k}^{(q)}\right\|^2
\leq P_{\rm T}.
\end{aligned}
\end{equation}
where $\Gamma_{k,q}$ denotes the post-cache-cancellation stream SINR and~\eqref{eq:finite_snr_bf} is solved using alternating LMMSE receive updates and SCA-based transmit updates~\cite{tolli2017multi,naseritehrani2026cache}.
For fixed $\Delta$, the same placement supports
$1\leq\beta\leq\min\{G,\lfloor L/\Delta\rfloor\}$; at high SNR,
the largest feasible $\beta$ maximizes the achievable DoF, 
whereas at finite SNR $\beta$ can be
selected to maximize $R_{\rm sym}$ by balancing spatial multiplexing
and same-profile interference.
Hence, the admissible operating points follow from the finite search
defined by \eqref{eq:ld_condition_clean_n} and \eqref{eq:integer_conditions_n}, while the CSI-dependent Tx/Rx optimization constitutes the main
online computation.\color{black}
\end{remarknum}

\section{Numerical Examples}
\label{section:Simulations}

%
%

We analyze the performance and complexity of the proposed scheme and compare it by considering
three baselines: the \textit{DoF-optimized} scheme in~\cite{naseritehrani2026cache}, which enables flexible DoF selection and achieves a large maximum DoF at the cost of extremely large subpacketization, the \textit{Lin-MIMO} scheme in~\cite{salehi2021MIMO}, which attains a large DoF value with linearly growing subpacketization but is applicable only under restrictive system parameters (i.e., $\frac{L}{G} \ge t$ and integer $\frac{L}{G}$), and a conventional multi-user MIMO (\textit{MU-MIMO}) scheme, where caching provides only local gain 
under a round-robin user scheduling.

Table~\ref{tab:theta_dof_low_vs_legacy} compares
the proposed and DoF-optimized schemes. 
The results highlight a key advantage of the proposed design: it enables flexible selection of the operating point $(\Omega,\beta)$, and hence the achievable DoF, while maintaining a 
lower subpacketization level.
At comparable DoF, the proposed scheme reduces
subpacketization by several orders of magnitude relative to the
DoF-optimized scheme.
For example, achieving DoF values in the range $28$–$36$ requires subpacketization on the order of $10^2$–$10^3$ in the proposed scheme, whereas the DoF-optimized scheme requires subpacketization on the order of $10^8$–$10^9$. This gap makes the DoF-optimized design impractical even for moderate system sizes, while the proposed scheme remains implementable and scalable. 
The reduced subpacketization might come at the cost of slightly reduced maximum DoF: with the DoF-optimized design, the achievable DoF reaches $38$ at $(\Omega,\beta)=(19,2)$, versus a cap of $36$ for the proposed scheme.
\color{black} 
For the same setup, MIMO-PDA~\cite{huang2026placement} 
achieves $(\mathrm{DoF},\Theta)\!=\!\!(38,7980)$ under its structural 
constraints, using the full stream order $\beta=G$. The resulting DoF-subpacketization 
tradeoff is further illustrated, e.g., for 
$(K,L,G,t)\!=\!\!(24,27,9,4)$, where the DoF-optimized, 
MIMO-PDA, and proposed schemes achieve $(\mathrm{DoF},\Theta)\!=\!\!(64,8.24\times10^{7})$, $(63,1.64\times10^{7})$, and $(60,756)$, respectively.
Hence, MIMO-PDA
maximizes DoF under its one-shot-ZF constraints while reducing
subpacketization, whereas~\cite{naseritehrani2026cache,
naseritehrani2026cacheasymmetric} and the proposed scheme retain
flexible $\beta$ selection, 
broadening the operating range and supporting finite-SNR rate optimization.\color{black}
\begin{table}[b]
\centering
\vspace{1.6mm}
\caption{Subpacketization comparison, proposed vs DoF-optimized schemes, $(L,G,\gamma) = (13,2,0.5)$, and $K=24$.
}
\label{tab:theta_dof_low_vs_legacy}
\vspace{-1mm}
\footnotesize
\setlength{\tabcolsep}{6pt}
\renewcommand{\arraystretch}{1.15}
\setlength{\arrayrulewidth}{0.4pt}
\resizebox{.81\columnwidth}{!}{%
\begin{tabular}{|c|c|c|}
\hline
\multicolumn{1}{|c|}{\diagbox[width=10.4em]{\hspace{-2mm}$(\Omega,\beta)\!\!\to\!\!\mathrm{DoF}$}{\vspace{1mm}$\Theta$}}
& \textbf{Proposed} & \textbf{DoF-optimized~\cite{naseritehrani2026cache}} \\
\hline
$(14,1)\!\to\!14$ & $924$               & $2.97\times 10^{7}$ \\
\hline
$(15,1)\!\to\!15$ & $70$                & $1.49\times 10^{8}$ \\
\hline
$(16,1)\!\to\!16$ & $20$                & $4.46\times 10^{8}$ \\
\hline
$(18,1)\!\to\!18$ & $6$                 & $1.25\times 10^{9}$ \\
\hline
$(24,1)\!\to\!24$ & $2$                 & $2.70\times 10^{6}$ \\
\hline
$(14,2)\!\to\!28$ & $1848$ & $5.95\times 10^{7}$ \\
\hline
$(15,2)\!\to\!30$ & $140$               & $2.97\times 10^{8}$ \\
\hline
$(16,2)\!\to\!32$ & $40$                & $8.92\times 10^{8}$ \\
\hline
$(18,2)\!\to\!36$ & $12$                & $2.50\times 10^{9}$ \\
\hline
\end{tabular}
}
\end{table}
Fig.~\ref{fig:Plot_SOTA_MC_UC_MC_MS_WSA_t2} illustrates the impact of operating-point selection on the finite-SNR performance of the proposed scheme. Optimized beamformers are employed in all simulations.
%
As observed, in the practical SNR regime, configurations with lower DoF can outperform higher-DoF operating points. This is because selecting smaller values of $(\Omega,\beta)$ reduces the number of simultaneously transmitted streams, thereby relaxing the interference suppression constraints and allowing more effective transmit and receive beamformer design. The resulting improvement in the effective SINR leads to higher symmetric rates despite the lower nominal DoF. In contrast, at high SNR, the performance becomes increasingly dominated by the DoF, and operating points with larger $(\Omega,\beta)$ achieve higher rates, as expected from asymptotic analysis. These results highlight a key advantage of the proposed scheme: the ability to adapt the operating point to the SNR regime, thereby achieving superior finite-SNR performance while still retaining high-DoF scalability. This finite-SNR advantage is also reflected in the comparison with the MU-MIMO scheme, where the proposed design consistently achieves higher symmetric rates over the considered SNR range.


Finally, Fig.~\ref{fig:Tx_side_scaling_DoF_bitsig_unconstrained} compares the maximum achievable DoF of the proposed and benchmark schemes
as the number of users $K$ increases, under explicit subpacketization constraints.
We consider three cases: (a) $\Theta < 10^{4}$, (b) $\Theta < 10^{6}$, and (c) unconstrained $\Theta$.
With unconstrained subpacketization, the DoF of the proposed scheme closely follows the DoF-optimized design, confirming the result in Lemma~\ref{lem:dof_gap}.
The Lin-MIMO scheme, however, benefits from linearly increasing subpacketization but is limited by the applicability condition $\frac{L}{G} \ge t$ and the integer constraint on $\frac{L}{G}$. To overcome the first constraint for large $K$ values, we need to implicitly split the system into multiple sub-networks, to reduce the effective $K$, and hence $t$. The second constraint is resolved by reducing the transmitter-side spatial multiplexing gain to $\alpha \le L$. Both tricks, however, affect the achievable DoF of the Lin-MIMO design, as is apparent in Fig.~\ref{fig:Tx_side_scaling_DoF_bitsig_unconstrained}. 
As the subpacketization constraint becomes active, the DoF-optimized approach loses its DoF advantage very quickly, and the proposed design performs better than both competing designs. Under all conditions, schemes benefiting from the CC gain outperform 
MU-MIMO.
Overall, 
Fig.~\ref{fig:Tx_side_scaling_DoF_bitsig_unconstrained} reveals a key scalability advantage of the proposed scheme: it maintains a favorable DoF scaling with $K$ while respecting the subpacketization constraints, owing to its flexible, low-complexity operating-point design. 
This demonstrates 
a superior trade-off between DoF scalability and implementation complexity, making it well-suited for large-scale systems under practical constraints.
 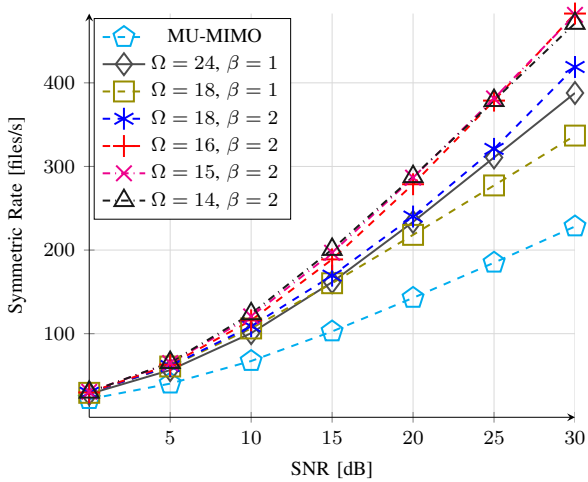
\begin{figure}[t]
    \centering

   \resizebox{.9\columnwidth}{!}{%

    \begin{tikzpicture}

    \begin{axis}
    [
    axis lines = center,
    xlabel near ticks,
    xlabel = \smaller {SNR [dB]},
    ylabel = \smaller {Symmetric Rate [files/s]},
    ylabel near ticks,
    ymin = 0,
    xmax = 30,
    legend style={
        nodes={scale= 0.99, transform shape},
        at={(0,1)}, 
        anchor=north west,
        draw=black, 
        outer sep=2pt, 
        font=\smaller, 
    },
    ticklabel style={font=\smaller},
    grid=both,
    major grid style={line width=.2pt,draw=gray!30},
    ]

        \addplot 
    [dashed, line width=0.8pt, mark = pentagon, mark size=4.4,  mark options={solid, fill=cyan!100} , cyan!100]
    table[y=MUMIMOomg13beta1,x=SNR]{Figs/Symmetric_rate_Prop_SOTA_Comparison.tex};
    \addlegendentry{\footnotesize MU-MIMO}

      \addplot 
    [, line width=0.8pt, mark = diamond, mark size=4.4,  mark options={solid, fill=black!70} , black!70]
    table[y=SCMIMOCComg24beta1,x=SNR]{Figs/Symmetric_rate_Prop_SOTA_Comparison.tex};
    \addlegendentry{\footnotesize $\Omega=24$, $\beta=1$}

    \addplot
    [dashed, line width=0.8pt,mark = square, mark size=4.0 ,  mark options={solid, fill=olive!100}, olive!100]
    table[y=SCMIMOCComg18beta1,x=SNR]{Figs/Symmetric_rate_Prop_SOTA_Comparison.tex};
    \addlegendentry{\footnotesize $\Omega=18$, $\beta=1$} 

     \addplot
    [dashed, line width=0.8pt,mark = asterisk, mark size=4.4,  mark options={solid, fill=blue!100}, blue!100]
    table[y=SCMIMOCComg18beta2,x=SNR]{Figs/Symmetric_rate_Prop_SOTA_Comparison.tex};
    \addlegendentry{\footnotesize $\Omega=18$, $\beta=2$} 

   \addplot
    [dashed, line width=0.8pt,mark = +, mark size=4.4,  mark options={solid, fill=red!100}, red!100]
    table[y=SCMIMOCComg16beta2,x=SNR]{Figs/Symmetric_rate_Prop_SOTA_Comparison.tex};
    \addlegendentry{\footnotesize $\Omega=16$, $\beta=2$} 
    
     \addplot
    [dash dot, line width=0.8pt, mark = x, mark size=4.4,  mark options={solid, fill=magenta!100} , magenta!100]
    table[y=SCMIMOCComg15beta2,x=SNR]{Figs/Symmetric_rate_Prop_SOTA_Comparison.tex};
    \addlegendentry{\footnotesize $\Omega=15$, $\beta=2$}

    \addplot
    [dash dot, line width=0.8pt, mark = triangle, mark size=4.4,  mark options={solid, fill=black!90} , black!90]
    table[y=SCMIMOCComg14beta2,x=SNR]{Figs/Symmetric_rate_Prop_SOTA_Comparison.tex};
    \addlegendentry{\footnotesize $\Omega=14$, $\beta=2$}

    \end{axis}

    \end{tikzpicture}
    }
    \vspace{-2mm}
    \caption{Symmetric rate of the proposed scheme, $(L,G,\gamma)=(13,2,0.5)$.
    }
    \label{fig:Plot_SOTA_MC_UC_MC_MS_WSA_t2}
\end{figure}


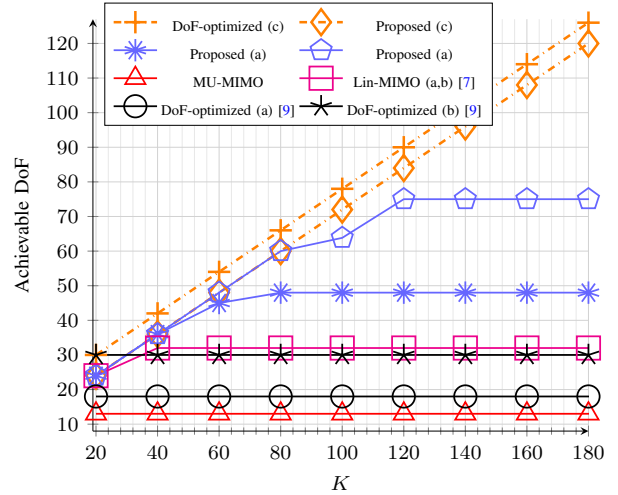
\begin{figure}[t]
    \centering

    \resizebox{.92\columnwidth}{!}{%

    \begin{tikzpicture}

    \begin{axis}
    [ 
      empty line=none,
    axis lines = center,
    xlabel near ticks,
    xlabel = \smaller {$K$},
    ylabel = \smaller {Achievable DoF},
    ylabel near ticks,
    ymin =8,
    ymax = 127,
    xmin = 19,  
    xmax = 180,
    xtick={20,40,60,80,100,120,140,160,180,200,220},
    ytick={10,20,30,40,50,60,70,80,90,100,110,120,140},
    legend style={
    nodes={scale=0.58, transform shape},
    at={(.415,0.73)}, 
    anchor=south, 
    draw=black,
    outer sep=2pt,
    font=\normalsize,
    legend columns=2, 
    scale=0.7, 
    text height=3.0ex, text depth=.25ex 
    },
    ticklabel style={font=\smaller},
    grid=both,
    minor tick num = 4, 
    major grid style = {line width=0.35pt,draw=gray!35},
    minor grid style = {line width=0.15pt,draw=gray!18}
    ]

    \addplot
    [dash dot,mark = +, line width=1.pt, mark size=4.4,mark options={solid, fill=orange!100!black} , orange!100]
    table[y=MIMOCC,x=K]{Figs/uncosntrained_MIMO_CC_L16_G6_gamma_0.1.tex};
    \addlegendentry{\normalsize DoF-optimized (c) 
    };

  \addplot
    [dash dot,mark = diamond, line width=1.pt, mark size=5.1,mark options={solid, fill=orange!100!black} , orange!100]
    table[y=SCMIMOCC,x=K]{Figs/uncosntrained_MIMO_CC_L16_G6_gamma_0.1.tex};
    \addlegendentry{\normalsize Proposed (c) 
    };

    \addplot
    [mark = 10-pointed star, line width=.7pt, mark size=4.4,mark options={solid, fill=blue!60} , blue!60]
    table[y=SCMIMOCC,x=K]{Figs/sig_Theta_max_1e4.tex};
    \addlegendentry{\normalsize Proposed (a) 
    };
    
    \addplot
    [mark = pentagon, line width=.7pt, mark size=4.4,mark options={solid, fill=blue!60} , blue!60]
    table[y=SCMIMOCC,x=K]{Figs/sig_Theta_max_1e6.tex};
    \addlegendentry{\normalsize Proposed (b) 
    };




   \addplot
    [mark = triangle, line width=.7pt, mark size=4.6 ,mark options={solid, fill=red!100}, red!100]
    table[y=MUMIMO,x=K]{Figs/sig_Theta_max_3e3.tex};  
    \addlegendentry{\normalsize MU-MIMO 
    };

    \addplot
    [mark = square, line width=.7pt, mark size=4.6,mark options={solid, fill=magenta!100} , magenta!100]
    table[y=MIMOCCWSA,x=K]{Figs/sig_Theta_max_1e6.tex};
    \addlegendentry{\normalsize Lin-MIMO (a,b)~\cite{salehi2021MIMO} 
    };

    \addplot
    [mark = o, line width=.7pt, mark size=4.4,mark options={solid, fill=black!80}, black!100]
    table[y=MIMOCC,x=K]{Figs/sig_Theta_max_1e4.tex};
    \addlegendentry{\normalsize DoF-optimized (a)~\cite{naseritehrani2026cache} 
    };


        \addplot
    [mark = star, line width=.7pt, mark size=4.4,mark options={solid, fill=black!80}, black!100]
    table[y=MIMOCC,x=K]{Figs/sig_Theta_max_1e6.tex};
    \addlegendentry{\normalsize DoF-optimized (b)~\cite{naseritehrani2026cache}
    };







    
    

    \end{axis}

    \end{tikzpicture}
    }
    \vspace{-2mm}
    \caption{DoF scaling versus $K$, under subpacketization limits: (a) $\Theta < 10^{4}$, (b) $\Theta < 10^{6}$, and (c) unconstrained $\Theta$. $(L,G,\gamma) = (16,6,0.1)$.
    }
    \label{fig:Tx_side_scaling_DoF_bitsig_unconstrained}
 
\end{figure}

\section{Conclusion}
\label{sec:conclusions}
A low-subpacketization MIMO-CC framework with flexible operating-point
selection is proposed. By extending the shared-cache paradigm to
multi-antenna receivers through a virtual 
BC decomposition,
the scheme substantially reduces subpacketization while retaining linear
decodability and flexibility across configurations. Under practical
subpacketization constraints, results show improved finite-SNR performance
and favorable DoF scaling relative to existing approaches.

\bibliographystyle{IEEEtran}
\bibliography{conf_short,IEEEabrv,references,whitepaper}
 
\end{document}